\documentclass{llncs}

\usepackage{sam-barr}
\usepackage[textsize=scriptsize]{todonotes}
\usepackage{wrapfig}

\begin{document}

\title{All Subgraphs of a Wheel are 5-Coupled-Choosable%
\thanks{Research of TB supported by NSERC.}}
\author{Sam Barr \and Therese Biedl}
\institute{David R. Cheriton School of Computer Science,
University of Waterloo, Waterloo, Canada. \email{\{s4barr,biedl\}@uwaterloo.ca}}
\date{\today}

\maketitle

\begin{abstract}
    A wheel graph consists of a cycle along with a center vertex
    connected to every vertex in the cycle.
    In this paper we show that every subgraph of a wheel graph
    has list coupled chromatic number at most 5, and this coloring
	can be found in linear time.
    We further show that `5' is tight for every wheel graph
    with at least 5 vertices, and briefly discuss possible generalizations
	to planar graphs of treewidth 3.
\end{abstract}

\section{Introduction}
\label{sec:intro}

In this paper we study the problem of \emph{coupled choosability},
the problem of finding a valid coloring given list assignments to every vertex
and face of a planar graph. The problem is of great relevance to list coloring
1-planar graphs, as list coupled coloring a planar graph corresponds to
list coloring an optimal 1-planar graph.
(Detailed definitions will be given in Section~\ref{sec:definitions}.)
Wang and Lih \cite{coupled-choosability-plane-graphs} 
show that every planar graph
is 7-coupled-choosable, and hence every optimal 1-planar graph is 7-choosable.
%It is an open problem whether this bound holds for all 1-planar graphs.
They further show
that maximal planar graphs are 6-coupled-choosable,
planar graphs of maximum degree 3 are 6-coupled-choosable, and
outerplanar graphs (and more generally, all $K_4$-minor free graphs)
are 5-coupled-choosable.

\iffalse
It is not hard to characterize $k$-coupled-choosability for small
values of $k$. Trivially, only the empty graph is 1-coupled-choosable,
and (among the connected graphs) only the single-vertex graph is 2-coupled-choosable.
%It is not hard to see that a graph $G$ is 2-coupled-choosable exactly when
%$G$ consists solely of isolated vertices. Moreover, 
A connected graph is 3-coupled-choosable
if and only if it is a tree, for at any edge $e$ incident to a cycle
we would require 4 colors for the endpoints of $e$ and the two
(distinct) incident faces of $e$.
%forest; a forest is clearly 3-coupled-choosable,
%and if an edge is on the boundary of 2 distinct faces, then 4 colors
%are required to color the two vertices and two faces adjacent to the edge.
\fi

The result by Wang and Lih settles the coupled choosability
for planar partial 2-trees (which are the same as $K_4$-minor free graphs).
Initially wishing to investigate the coupled choosability
of planar partial 3-trees, in this paper we investigate the coupled choosability
of wheel graphs and their subgraphs. 
% TB: this will be said later---either define everything here, or nearly nothing
%A wheel graph consists of a cycle along with a center vertex connected to
%every vertex in the cycle; every wheel graph is a partial 3-tree.
In Theorem~\ref{main_theorem}, we show that any subgraph of a wheel is
5-coupled-choosable, and the coloring can be found in linear time.
(Prior papers such as \cite{coupled-choosability-plane-graphs} did not
address the run-time of finding their colorings; it can clearly be
done in polynomial time by following the steps of their proof but
linear time is not obvious.)
In
Theorem~\ref{characterization_theorem}, we characterize the coupled choosability
of wheel graphs by showing that 5 is tight for wheel graphs with at least
5 vertices. In the last section of the paper, we touch upon how these results
could be relevant in finding the coupled choosability of planar partial 3-trees.

As for related results,
the (non-coupled) choosability of wheel graphs was characterized in
a different paper by Wang and Lih \cite{2005-wang-lih-halin}: wheels of even order
have list chromatic number 4, while wheels of odd order have list chromatic number 3.
This stands in contrast to our result, as the parity of the number of vertices
in the graph does not affect the coupled choosability of wheel graphs.
Wang and Lih also show that Halin graphs that are not wheels have list chromatic number 3,
while in Theorem~\ref{counterexample} we prove the existence of
a Halin graph that is not 5-coupled-choosable (in fact, it is not 5-coupled-colorable).

Our paper is structured as follows:
In Section~\ref{sec:definitions} we will go over the necessary definitions and terminology for graphs
and graph coloring. In Section~\ref{sec:wheel} we investigate the coupled choosability of wheel graphs.
In Section~\ref{sec:operations} we examine how coupled choosability behaves under certain graph operations.
In Section~\ref{sec:subgraph} we extend our analysis of wheel graphs to subgraphs of wheels,
along with lower-bounding the coupled choosability of wheel graphs.
In Section~\ref{sec:tw3} we go over several possible extensions to our results, in particular
some conjectures about the coupled-choosability of planar partial 3-trees.

\section{Definitions}
\label{sec:definitions}

We assume basic familiarity with graph theory (see \cite{diestel}).
In this paper all graphs are finite and connected.

The \emph{complete graph} $K_4$ consists of four vertices and all possible
edges between them.  
A \emph{subdivision} of a graph $G$ is formed by repeatedly
taking some edge $uv \in E(G)$, removing $e$ from $G$, adding a new vertex $x$,
and adding edges $ux$ and $xv$.  A graph is called \emph{$K_4$-minor free} if
none of its subgraphs is a subdivision of $K_4$.  

We recall that a graph $G$ is called \emph{planar} if it can be
drawn in the plane without edges crossing, and \emph{plane} if
a specific planar drawing $\Gamma$ is given.  
%If such a graph $G$ is given with a prescribed drawing it is known
%as a \emph{plane graph}. The edges of this drawing subdivide the plane
%into regions called \emph{faces}. 
The maximal regions of $\mathbb{R}\setminus \Gamma$ are called \emph{faces};
%There is a unique face whose region is unbounded, this face 
the unbounded region
is known as the \emph{outer face} and all other faces are \emph{inner faces}.
% TB: moved outerplanar to here; it was sticking out weirdly at the end
An \emph{outerplanar} graph is a graph that can be drawn in the plane such that
every vertex is on the outer face; such a graph is $K_4$-minor free.
A \emph{bigon} is a face that is bounded 
by two duplicate edges between a pair
of vertices.
%, then we refer to the face as a \emph{bigon}.
For a plane graph $G$, we use $V(G)$, $E(G)$, and $F(G)$ to denote the set of
\emph{vertices}, the set of \emph{edges}, and the set of \emph{faces} of $G$,
respectively.
The \emph{dual graph} $G^*$ of a plane graph is obtained by exchanging
the roles of vertices and faces, i.e., $G^*$ has a vertex for every face
of $G$, and an edge $(f_1,f_2)$ for every common edge of the two corresponding
faces $f_1,f_2$ in $G$.

A \emph{list assignment} is a map $L$ that assigns a set of \emph{colors}
for each vertex or face in $V(G) \cup F(G)$. A
\emph{coupled coloring with respect to $L$}
is a map $c$ such that $c(x) \in L(x)$ for every $x \in V(G) \cup F(G)$,
% TB: avoid side-sentences (and especially avoid adjacent math-formulas without
% English to separate them.
and $L(x)\neq L(y)$ for incident or adjacent elements $x,y \in V(G) \cup F(G)$.
%, $L(x) \neq L(y)$.
If such a map $c$ exists, then we say that $G$ is \emph{$L$-coupled-choosable}.
If $G$ is $L$-coupled-choosable for every $L$ such that
$|L(x)| =k$ for every $x \in V(G) \cup F(G)$, then we say that $G$
is \emph{$k$-coupled-choosable}. The smallest integer $k$ such that
such that $G$ is $k$-coupled-choosable is called the \emph{list coupled chromatic number}
of $G$ and denoted $\chi^L_{vf}(G)$. 
Observe that a list coupled coloring of a graph $G$ implies a list coupled coloring
of the dual graph $G^*$, since the roles of the vertices and the faces is exchanged
but incidences/adjacencies stay the same. Hence, we have
$\chi^L_{vf}(G) = \chi^L_{vf}(G^*)$.

A natural way to express the list coupled chromatic number is to
define a new graph $X(G)$ with vertices for all vertices and faces of $G$
and edges whenever the vertices and faces $G$ are adjacent/incident.
This graph $X(G)$ is \emph{1-planar}, i.e., can be drawn in the plane
with at most one crossing per edge.  In fact, if $G$ is 3-connected
then $X(G)$ is an \emph{optimal
1-planar graph}, i.e., it is simple and has the maximum-possible $4n-8$ edges.
(All optimal 1-planar graphs can be obtained in this fashion \cite{Schumacher}.)
A coupled coloring of $G$
corresponds to a \emph{vertex coloring} of $X(G)$, i.e., a coloring
of the vertices such that adjacent vertices have different colors.  When restricting
a vertex coloring to given lists $L$,  then the respective
terms are \emph{$L$-choosable}, \emph{$k$-choosable}, and the
\emph{list chromatic number} $\chi^L(X)$.

The \emph{wheel graph} $W_n$ is formed by starting with a cycle $C_{n-1}$
on $n-1$ vertices (the \emph{outer cycle}),
adding a \emph{center vertex} inside the cycle
and adding a \emph{spoke-edge} from the center vertex to every vertex on the cycle.
We will label the center vertex and the outer face of the wheel graph as $x_0$
and $f_0$, respectively. We further label the vertices in the outer cycle
as $x_1, \ldots, x_{n-1}$, and label the inner faces as $f_1, \ldots, f_{n-1}$
such that $x_i$ is incident to $f_i$ and $f_{i+1}$ for $1 \leq i < n-1$,
and $x_{n-1}$ is adjacent to $f_{n-1}$ and $f_1$ (see Figure~\ref{graph:9}).

\section{Coupled Choosability of Wheel Graphs}
\label{sec:wheel}

In order to prove the desired result for all subgraphs of the wheel graph,
we first determine the coupled choosability of the wheel graph itself.
It will be helpful to recall the following result relating
the choosability of a graph to the maximum degree; it is an analogue
to Brook's theorem and similarly upper-bounds the chromatic number
of a graph by its maximum degree.

\begin{lemma}
    \label{erdos}
    (Erd\H{o}s, Rubin, and Taylor \cite{erdos-choosability-in-graphs})
    Let $G$ be a connected graph that is neither an odd cycle nor a complete
    graph. Then $G$ is $\Delta(G)$-choosable.
\end{lemma}

Our main result in this section is:

\begin{lemma}
    \label{main_lemma}
    Every wheel graph $W_n$, $n \geq 4$, is 5-coupled-choosable.
\end{lemma}
\begin{proof}
    For $n=4$, $W_4$ is the complete graph $K_4$.
    Wang and Lih \cite{coupled-choosability-plane-graphs}
    proved that $\chi^L_{vf}(K_4)=4$, so we assume $n \geq 5$. 
    Let $L$ be a color assignment for $W_n$ such that
    $|L(y)|=5$ for every $y \in V(W_n) \cup F(W_n)$.
    Our goal is to find a coupled coloring with respect to $L$.
    Since $x_0$ and $f_0$ are both adjacent to all remaining
	vertices, we will color them first and then color the rest of $X(W_n)$.
	We will use $X_n$ as a shortcut for $X(W_n)\setminus \{x_0,f_0\}$.
	%We will use $X_n$ as a shortcut for $X(W_n)\setminus \{x_0,f_0\}$, and
	%recall that this consists of the following:
    %\begin{itemize}
    %    \item The vertices of $X_n$ correspond to vertices $x_1, \ldots, x_{n-1}$
    %        and faces $f_1, \ldots, f_{n-1}$ of $W_n$.
    %    \item There is an edge $x_ix_j$ if the vertices $x_i$ and $x_j$
    %        are adjacent in $W_n$.
    %    \item There is an edge $f_if_j$ if the faces $f_i$ and $f_j$
    %        are adjacent in $W_n$.
    %    \item There is an edge $x_if_j$ if the vertex $x_i$ is incident
    %        to the face $f_j$ in $W_n$.
    %\end{itemize}
    Observe that $|V(X_n)|=2n-2$
    and that $X_n$ is 4-regular (see Figure~\ref{graph:9}).
    Furthermore, it suffices to find a vertex-colouring of $X_n$ with
    respect to $L$, plus two suitable colors in $L(x_0)$ and
    $L(f_0)$ for $x_0$ and $f_0$.  We have two cases:

    \begin{figure}[ht]
        %        \centering
        % equal out the spacing
        \hspace*{\fill}
        \includegraphics[scale=0.45,page=2]{x_graph.pdf}
        \hspace*{\fill}
        \includegraphics[scale=0.45,page=1]{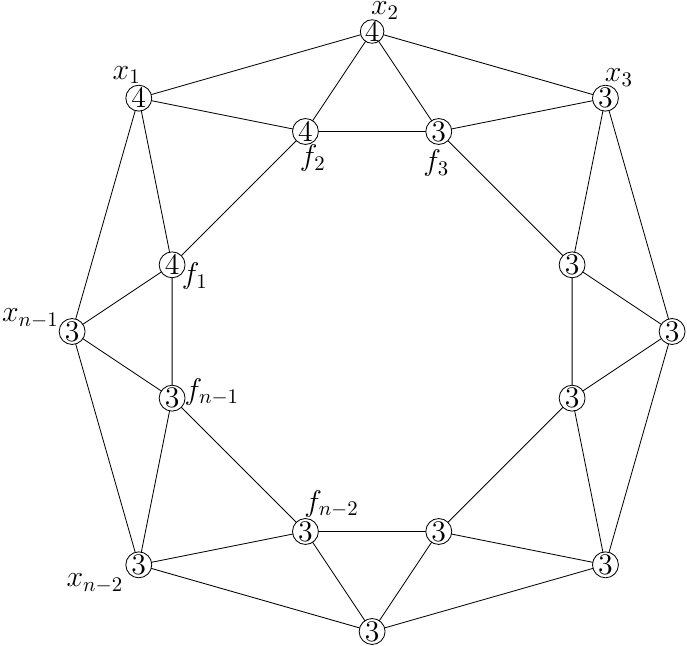}
        \hspace*{\fill}
        \caption{The graph $W_9$ (left) and $X_9$ (right).  Circled numbers indicate a
        lower bound on the list-length in $L'$.}
        \label{graph:9}
    \end{figure}

    \medskip\noindent\textbf{Case 1:} $L(x_0) \cap L(f_0) \neq \emptyset$.
    Let $a \in L(x_0) \cap L(f_0)$, and assign color $a$ to $x_0$ and $f_0$.
% TB: the following was essentially said earlier
%    It remains to color the graph $X_n$ without using $a$.
    Observe that $|L(y) \setminus\{a\}| \geq 4$ for every $y \in V(X_n)$
    and $X_n$ has maximum degree 4. Moreover $|X_n|=2n-2$ is even,
    so $X_n$ is not an odd cycle. Also $x_1$ and $x_3$ are not adjacent by
	$n\geq 5$, so $X_n$
    is not a complete graph.
    Therefore, by Lemma~\ref{erdos}, we have a list coloring of the
    vertices of $X_n$ that only uses colors in $L \setminus \{a\}$, which
    in turn implies an $L$-list-coloring of the vertices and faces of $W_n$.

    \medskip\noindent\textbf{Case 2:} $L(x_0) \cap L(f_0) = \emptyset$.
    We find suitable colors for $x_0$ and $f_0$ by imitating the method used for $K_4$ 
    in \cite{coupled-choosability-plane-graphs} (but adapted here to 5 colors).
    Define \emph{color-pairs} $S := \setb{ \{a,b\} }{a \in L(x_0), b \in L(f_0)}$.
    By case-assumption $|S| = 25$. 

    We claim that $|\setb{s \in S}{s \subseteq L(y)}| \leq 6$
    for any $y \in V(X_n)$. 
    To see this, let $y \in V(X_n)$, and consider the disjoint sub-lists
    $L_1 := L(y) \cap L(x_0)$ and $L_2 := L(y) \cap L(f_0)$.
    Since $|L_1|+|L_2|\leq|L(y)|=5$, and $|L_1|$ and $|L_2|$ are integers,
    we have
    \[
        |\setb{s \in S}{s \subseteq L(y)}|
        = |L_1 \times L_2|
        = |L_1| \cdot |L_2|
        \leq 6.
    \]

    Therefore, color-pairs of $S$
    appear as subsets of lists in $X_n$ at most
    \[
        \sum_{y \in X_n} |\setb{s \in S}{s \subseteq L(y)}|
        \leq (2n-2) \cdot 6 = 12n - 12
    \]
    times. By $|S|=25$, some element $\{a',b'\}$ of $S$ appears at most
    \[
        \frac{12n-12}{25} < \frac{n-1}2
    \]
    times as a subset of a list in $X_n$.
    Color $x_0$ with $a'$ and $f_0$ with $b'$.
    For $y \in V(X_n)$, define $L'(y) := L(y) \setminus \{a', b'\}$.
    For any $y \in V(X_n)$, we have $3 \leq |L'(y)| \leq 5$.
    We call $y$ a \emph{3-vertex} if $|L'(y)|=3$ (this implies $\{a',b'\}\subset L(y)$),
    and a \emph{4-vertex} otherwise. 
%
%    If $|L'(y)| = k$, we refer to $y$ as a \emph{$k$-vertex}.
    From our choice of colors $a'$ and $b'$, we have
    \[
        \frac{|\setb{y \in V(X_n)}{\text{$y$ is a 3-vertex}}|}{|V(X_n)|}
        < \frac{(n-1)/2}{2n-2}
        = \frac14
    \]
    Therefore, more than three quarters of the vertices of $X_n$ are 4-vertices.
%    and 5-vertices. We will
%    assume that these are all 4-vertices (we can just delete one color from the
%    list of a 5-vertex to make it a 4-vertex).
    Consider the cyclic enumeration
    \[ \sigma := \angles{f_1, x_1, f_2, x_2, \ldots, f_{n-1}, x_{n-1}} \]
    of the vertices of $X_n$. Since strictly more than $\frac34|V(X_n)|$ of the
    vertices are 4-vertices, we have four consecutive 4-vertices
    in $\sigma$. Up to exchange of $f_i$ and $x_i$ and renumbering,
    we may assume that $f_1,x_1,f_2$, and $x_2$ are 4-vertices.
	Figure~\ref{graph:9}(right) illustrates the lower bounds on
	the size of $L'$.

	We next color $f_{n-1},x_{n-1}$ and $x_1$ and have two sub-cases.
    If $L'(f_{n-1})\cap L'(x_1) \neq\emptyset$, then color $f_{n-1}$
    and $x_1$ with the same color. 
    Otherwise, since $|L'(f_{n-1}) \cup L'(x_1)| \geq 7 > |L(f_1)|$,
    there are colors $p$ and $q$ for $f_{n-1}$ and $x_1$
    respectively such that at least one of them is not in $L(f_1)$, i.e., $|L(f_1) \cap \{p,q\}| \leq 1$.
    Pick these colors for $f_{n-1}$ and $x_1$. 
    In either case, two vertices adjacent to $x_{n-1}$ have been colored,
    and $|L'(x_{n-1})|\geq 3$, so $x_{n-1}$ will have at least one valid color left,
    and we pick this color for $x_{n-1}$.

    We now have colors $p,q$, and $r$ for $f_{n-1},x_1,$ and $x_{n-1}$
    (respectively) such that $|L'(f_1) \cap \{p,q,r\}| \leq 2$.
    Removing these colors from the lists of their neighbors produces
    new lists $L''$ such that

\noindent\begin{minipage}{0.55\linewidth}
    \begin{align*}
        |L''(f_1)| & = |L'(f_1) \setminus \{p,q,r\}| \geq 4-2 = 2 \\
        |L''(f_2)| & = |L'(f_2) \setminus \{q\}| \geq 4-1 = 3 \\
        |L''(x_2)| & = |L'(x_2) \setminus \{q\}| \geq 4-1 = 3 \\
        |L''(x_{n-2})| & = |L'(x_{n-2}) \setminus\{p,r\}| \geq 3-2 = 1 \\
        |L''(f_{n-2})| & = |L'(f_{n-2}) \setminus \{p\}| \geq 3-1 = 2 \\
        |L''(x_i)| & \geq 3 \quad (\text{for all } 3 \leq i \leq n-3) \\
        |L''(f_i)| & \geq 3 \quad (\text{for all } 3 \leq i \leq n-3)
    \end{align*}
The figure on the right illustrates these lower bounds on the list-lengths in $L''$.
\end{minipage}
\qquad
\begin{minipage}{0.4\linewidth}
        \includegraphics[scale=0.45,page=3]{x_graph.pdf}
\end{minipage}
\medskip

%\begin{figure}[ht]
%        \centering
%        \includegraphics[scale=0.45,page=3]{x_graph.pdf}
%        \caption{The graph $X'_n$ (solid). Dotted edges show the rest
%            of the graph $X_n$, along with a dashed edge to the vertex $f_1$
%            which remains to be colored. Numbers indicate lower bounds
%            on the list length in $L''$.
%        }
%        \label{xprime}
%\end{figure}

    Let $X_n' := X_n \setminus \{f_{n-1}, x_{n-1}, f_1, x_1\}$ ($X_n'$ is solid in the
	above figure) and color
    it with respect to list assignment $L''$.
    This is feasible since $X_n'$ is outerplanar % (see Figure~\ref{xprime}),
    and outerplanar graphs are 3-choosable even
    if the colors of two consecutive vertices on the outer face
    are fixed \cite{outer} (here we fix the colors for $x_{n-2}$ and $f_{n-2}$). 
    This colors all vertices except for $f_1$, but $|L''(f_1)| \geq 2$
    and $f_1$ has only one neighbor in $X_n'$, so we can give
    it a color not used by $f_2$.
    Therefore, we have a list vertex-coloring of $X_n$ that is
	compatible with the colors for $x_0,f_0$ chosen earlier
	and so implies a list coupled coloring of $W_n$.
\end{proof}

Note that this coloring can easily be found in linear time.    This is
obvious in Case~1 due to the linear-time result for Lemma~\ref{erdos}
\cite{erdoslinear}.   Determining the colors $a,b$ for Case~2 takes 
linear time since all list-lengths are constant, and then we mostly
appeal to list-coloring an outer-planar graph, which can be done in
linear time since outer-planar graphs are 2-degenerate.

\section{Coupled Choosability Under Graph Operations}
\label{sec:operations}

In Section~\ref{sec:subgraph} we seek to prove that all subgraphs of a wheel
are 5-coupled-choosable. In pursuit of this, we examine how various graph operations
affect the list-coupled-chromatic number. First, in contrast to
list-vertex-coloring, there is no clear relationship
between the list coupled chromatic number of a graph and the list coupled chromatic
number of its subgraphs. Indeed, a subgraph may have larger
list coupled chromatic number.

\begin{observation}
    \label{sub:observation}
    There exists a plane graph $G$ with subgraph $H \subseteq G$
    such that $\chi^L_{vf}(H) > \chi^L_{vf}(G)$
\end{observation}
\begin{proof}
Let $H$ be the graph obtained by deleting one edge of $K_4$;
    see Figure~\ref{subgraph}.
    From Theorem~10 of \cite{coupled-choosability-plane-graphs},
    we know that the graph $K_4$ is 4-coupled-choosable,
    i.e., $\chi^L_{vf}(K_4)=4$.
    But in graph $H$, %observe that 
	the incidences and adjacencies
    between $x_0,x_1,x_2,f_2$, and $f'$ form a $K_5$, and therefore
    $\chi^L_{vf}(H) \geq 5 > 4 = \chi^L_{vf}(K_4)$. 
%    \[ \chi^L_{vf}(H) \geq 5 > 4 = \chi^L_{vf}(K_4) \]
\end{proof}

\begin{figure}[ht]
    \hspace*{\fill}
    \includegraphics[scale=0.45,page=7]{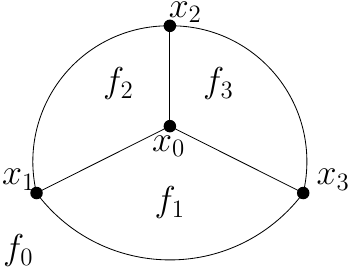}
    \hspace*{\fill}
    \includegraphics[scale=0.45,page=8]{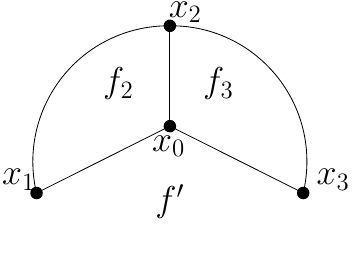}
    \hspace*{\fill}
    \caption{$K_4$ and subgraph $H$.
Observe that $\chi^L_{vf}(H)=5$ since it is outerplanar.}
    \label{subgraph}
\end{figure}

Other graph operations are better behaved in this respect.
For instance, there is a clear relationship between the coupled choosability
of some graph $G$ and the coupled choosability of any subdivision of $G$.

\begin{lemma}
    \label{subdivision_lemma}
    For any plane graph $G$, %be a plane graph such that $\chi^L_{vf}(G) \geq 5$.  Then
    any subdivision $H$ of $G$ is $\max\{5, \chi^L_{vf}(G)\}$-coupled-choosable.
\end{lemma}
\begin{proof}
    Let $L$ be a list assignment for $H$ such that $|L(x)|=\max\{5,\chi^L_{vf}(G)\}$
    for every vertex and face of $H$. We prove the statement
    by induction on the number of subdivisions performed on $G$
    to obtain $H$. If $H$ is the result of subdividing the edges of $G$ zero times,
    then $H=G$ and so trivially any $L$-coupled-coloring of $G$ is an
    $L$-coupled-coloring of $H$.

    Otherwise, $H$ was the result of performing $k+1$ subdivisions on $G$ for
	some $k\geq 0$.
    In particular, $H$ is the result of subdividing a single edge of some
    graph $H'$, where $H'$ was the result of performing $k$ subdivisions on $G$.
    Let $uv \in E(H')$ be the edge of $H'$ that was subdivided,
    and let $x$ be the vertex which was added.
    By the inductive hypothesis, $H'$ is $\max\{5,\chi^L_{vf}(G)\}$-coupled-choosable.
    Color the faces of $H$ and the vertices $V(H) \setminus \{x\}$ according
    to how they would be colored in $H'$. Then we only need to color
    the remaining vertex $x$. Note that $x$ has degree two with neighbors
    $u$ and $v$. Let $f_1$ and $f_2$ be the two faces adjacent
    to the edge $uv$ in $H'$. Then $u,v,f_1$, and $f_2$ are the only
    vertices and faces that are adjacent (respectively incident) to $x$.
    Hence, after coloring the vertices and faces from $H'$, $x$ still has at least
    $|L(x)| - 4 \geq 5 - 4 = 1$ 
%    \[ |L(x)| - 4 \geq 5 - 4 = 1 \]
    color left and can be colored.
\end{proof}

This implies another result. % that we will not need but find interesting.
% TB: actually, we can use this result for the triple edge below
For a planar graph $G$, subdividing an edge corresponds in the dual graph $G^*$
to duplicating edges to form bigons. Since $\chi^L_{vf}(G)=\chi^L_{vf}(G^*)$ we
therefore have:

\begin{corollary}
\label{cor:bigon}
    Let $G$ be a plane graph, and $H$ the result of duplicating
    some edges of $G$ to form bigons. 
    Then $H$ is $\max\{5, \chi^L_{vf}(G)\}$-coupled-choosable.
%	If $\chi^L_{vf}(G) \geq 5$,
%    then $H$ is $\chi^L_{vf}(G)$-coupled-choosable. If $\chi^L_{vf}(G) \leq 5$,
%    then $H$ is 5-coupled-choosable.
\end{corollary}

A similar result can also be had for adding a vertex of degree one to a graph.

\begin{lemma}
    \label{lemma:degree1}
    Let $G$ be a planar graph, and let $H$ be $G$ plus a new vertex of degree one.
    Then $H$ is $\max\{3, \chi^L_{vf}(G)\}$-coupled-choosable.
\end{lemma}
\begin{proof}
    Let $x$ be the new vertex, and let
    $L$ be a list assignment for $H$ such that $|L(y)| = \max\{3,\chi^L_{vf}(G)\}$
    for every vertex and face of $H$. Color the faces and vertices of
    $H-x$ according to how they would be colored in $G$. It remains to color $x$.
    Since $x$ is adjacent to only one vertex and incident to only one face in $H$,
    after coloring the vertices and face of $H-x$, $x$ still has at least
    $|L(x)| - 2 \geq 3 - 2 = 1$ 
%    \[ |L(x)| - 2 \geq 3 - 2 = 1 \]
    color left and can be colored.
\end{proof}

Note that for all three of the above lemmas, the coloring of $H$ can be found in
constant time, given a suitable coloring of $G$. 

Wang and Lih~\cite{coupled-choosability-plane-graphs} proved that all $K_4$-minor free
graphs are 5-coupled-choosable, but it is not clear whether their proof leads to a
linear-time algorithm to find the coloring.  With the above two results, such an
algorithm is immediate.

\begin{theorem}
    \label{k4:new:proof}
    All $K_4$-minor free graphs are 5-coupled-choosable, and the coloring can
	be found in linear time.
\end{theorem}
\begin{proof}
    It is known (see \cite{k4-minor-result}) that every $K_4$-minor free
    graph $G$ can be obtained from some tree $T$ via a series of duplicating edges,
    subdividing edges, and adding vertices of degree one. Then by
    Lemmas \ref{subdivision_lemma} and \ref{lemma:degree1}, Corollary \ref{cor:bigon},
    and the 3-coupled-choosability of trees, we have that $G$ is 5-coupled-choosable.

	To find the coloring efficiently, first split $G$ into its 2-connected
	components $C_1,\dots,C_d$ \cite{HT73}. Then run on each component $C_i$
	the algorithm 
	that recognizes so-called series-parallel graphs in linear  time \cite{VTL82}.
	Since 2-connected $K_4$-minor free graphs are series-parallel graphs,
	this algorithm will succeed on each $C_i$, and 
	following the trace of its execution one obtains how to construct $C_i$ from
	a single edge via a series of duplicating edges and subdividing edges.  Combining
	this with the tree of 2-connected components shows how to obtain $G$.  Since
	trees are trivially 3-coupled-colorable (choose a color for the unique face, then
	find a 2-coloring of the vertices), and each of our expansion steps takes constant
	time, we can find the coloring of $G$ in linear time.
\end{proof}

%Moreover, the decomposition of a $K_4$-minor free graph into a tree can be found
%in linear time, and hence we have that we can 5-coupled-choose a
%$K_4$-minor free graph in linear time (a result which can not be immediately obtained
%from the proof by Wang and Lih).

\section{Subgraphs of Wheels}
\label{sec:subgraphs}
\label{sec:subgraph}
%\todo[inline]{this section now is more "subgraphs of wheels and characterizing
%coupled choosability of wheels"}

We now turn to graphs that are subgraphs of wheels.  
As demonstrated in Observation \ref{sub:observation},
non-trivial work is required to demonstrate that any subgraph
of a wheel graph is also 5-coupled-choosable. The result comes quickly from the
results proved in the previous section.

\begin{theorem}
    \label{main_theorem}
    Let $G$ be a subgraph of a wheel graph $W_n$, $n \geq 4$.
    Then $G$ is 5-coupled-choosable and the coloring can be
	found in linear time.
\end{theorem}
\begin{proof}
    We examine several possibilities of the structure of $G$.

    \medskip\noindent\textbf{Case 1:} $G=W_n$. Then by Lemma \ref{main_lemma}
    $G$ is 5-coupled-choosable, and the coloring can be found in linear time.

    \medskip\noindent\textbf{Case 2:} $G$ is the result of deleting at least one edge or
    vertex of $W_n$ that is on the outer face.
	Then $G$ is outerplanar and therefore $K_4$-minor free, and so
    by Theorem~\ref{k4:new:proof}, $G$ is 5-coupled-choosable and the coloring
	can be found in linear time.
%\todo{TB: Can we actually find this colouring in linear time?  I suspect that yes
%but that's not trivial.  For the IWOCA paper I'd ignore this issue, but for your
%thesis you might want to review the proof and then conclude that this is linear.}

    \medskip\noindent\textbf{Case 3:} $G$ is the result of removing the center vertex
    of $W_n$. Then $G=C_{n-1}$ is outerplanar and (as in the previous case) 5-coupled
	choosable.
    
    \medskip\noindent\textbf{Case 4:} None of the above. Then all vertices of $W_n$
    belong to $G$, but we deleted some edges which were not on the outer face.
    So $G$ is the result of deleting some of the spoke-edges incident
    to the center vertex. 
%    If 1 or 2 spokes remain, then $G$ does not contain a $K_4$-minor,
%    and so by Theorem~14 in \cite{coupled-choosability-plane-graphs}
%    $G$ is 5-coupled-choosable.
    If at most two spokes remain, then $G$ has at most 3 faces and therefore
    is $K_4$-minor free, and hence is 5-coupled-choosable by Theorem~\ref{k4:new:proof}.
	If at least three spokes remain, then
    $G$ is a subdivision of some $W_k$
    for $k \geq 4$. By Lemmas~\ref{main_lemma} and \ref{subdivision_lemma}
    $G$ is 5-coupled-choosable, and we can find the coloring in linear time
	since we can detect all subdivision-vertices by scanning for vertices
	in linear time.
\end{proof}

%Following the steps of our proof, one can easily verify that we can find
%the $L$-coupled-coloring in linear time. 

Having established an upper bound on the list coupled chromatic number of
wheel graphs in Lemma \ref{main_lemma}, one might wonder whether this bound
is tight or not.
In \cite{coupled-choosability-plane-graphs}, it is shown that the graph
$K_4=W_4$ is 4-coupled-choosable. In fact, this is the only wheel
graph that is 4-coupled-choosable. For all other wheel graphs, the bound
of 5-coupled-choosability is tight.

\begin{theorem}
    \label{characterization_theorem}
    $\chi^L_{vf}(W_n) = 5$, for $n \geq 5$.
\end{theorem}
\begin{proof}
    From Lemma \ref{main_lemma}, we know that all wheel graphs are
    5-coupled-choosable. It remains to show that they are not
    4-coupled-choosable for $n \geq 5$.

    For $n=5,6$, we consider the list assignment $L$ such that
    $L(y) = \{1,2,3,4\}$ for every $y \in V(W_n) \cup F(W_n)$.
	(So these graphs are not even 4-coupled-colorable.)
    Assume for contradiction that we have an $L$-coupled-coloring $c$ of $W_n$.
% TB: When you're writing English (not math) then short numbers should be written
% out as words.  Not sure where the boundary is, but definitely for 1-6.
    If $c(x_0) \neq c(f_0)$, then this leaves two colors for coloring the triangle
    $x_1,f_1,f_2$ in $X_n$, impossible. Hence $c(x_0)=c(f_0)$, say they
    are both colored 4.
    Then we have an $L'$-coloring of $X_n$ with lists $L'(y) := L(y) \setminus \{4\} = \{1,2,3\}.$ 

%    \[ L'(y) := L(y) \setminus \{4\} = \{1,2,3\}. \]
    
    Observe that for $X_5$ and $X_6$, any putative $L'$-coloring would be unique
    up to renaming the colors, since once we have colored one triangle, every
    other vertex can be reached via a sequence of triangles.
    One verifies that for these graphs (and indeed every $X_k$ where $k-1$
    is not divisible by 3), attempting such a 3-coloring leads to a contradiction
    (see Figure~\ref{x5x6}).
    This proves Theorem~\ref{characterization_theorem} for $n=5,6$.

    For $n \geq 7$, we construct a list assignment $L$
    such that $W_n$ is not $L$-coupled-choosable. Set 
    $L(x_0)=\{1,2,3,4\}$ and $L(f_0)=\{5,6,7,8\}$.
    We further define:
    \begin{align*}
        L(f_1) = L(x_1) = L(f_2) & = \{1,2,5,6\} \\
        L(x_2) = L(f_3) = L(x_3) & = \{1,2,7,8\} \\
        L(f_4) = L(x_4) = L(f_5) & = \{3,4,5,6\} \\
        L(x_5) = L(f_6) = L(x_6) & = \{3,4,7,8\} \\
    \end{align*}
    Observe that each of these triples forms a triangle in $X_n$,
    and for any $a \in \{1,2,3,4\}$ and $b \in \{5,6,7,8\}$,
    one of these triangles has colors $\{a,b,x,y\}$ for some colors $x,y$.
    Assume for contradiction that we have an $L$-coupled-coloring $c$ of $W_n$.
    Up to symmetry, assume $c(x_0)=1$ and $c(f_0)=5$. But then
    $f_1$, $x_1$, and $f_2$ have two colors left, and therefore cannot be colored,
    a contradiction.
\end{proof}

    \begin{figure}[ht]
        \hspace*{\fill}
        \includegraphics[scale=0.45,page=9]{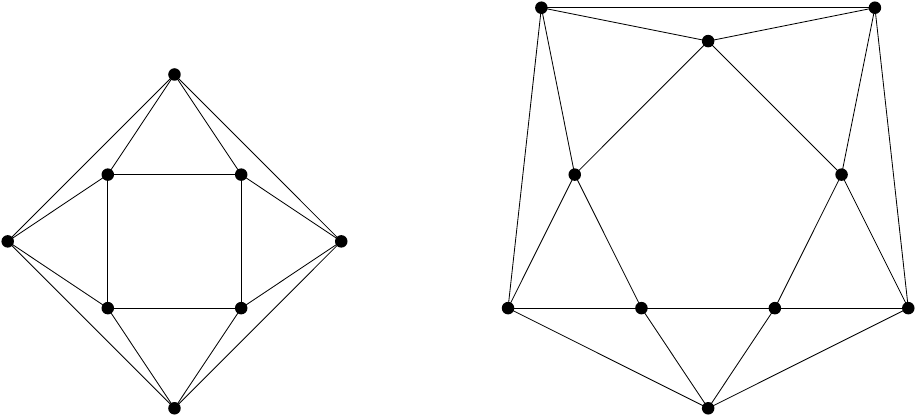}
        \hspace*{\fill}
        \caption{The graphs $X_5$ (left) and $X_6$ (right).}
        \label{x5x6}
    \end{figure}
%\todo{try illustrating the propagation process via labels}

With this, we can characterize the coupled choosability
of wheel graphs.

\begin{corollary}
    For a wheel graph $W_n$, we have $\chi^L_{vf}(W_n)=\min\{5,n\}$.
%    \[
%        \chi^L_{vf}(W_n) = \begin{cases}
%            4 & n = 4 \\
%            5 & n \geq 5
%        \end{cases}
        %= \min(5, n)
%    \]
\end{corollary}

\section{Towards Partial 3-trees}
\label{sec:tw3}

Our investigation of wheel graphs was motivated by wanting to determine
the coupled choosability number of planar partial 3-trees.  To define 
these, we first define {\em Apollonian networks} recusively as follows.  
A triangle is an Apollonian network.  If $G$ is an Apollonian network,
and $f$ is a face of $G$ (necessarily a triangle) that is not the outer-face,
then the graph obtained by stellating face $f$ is also an Apollonian network.
Here {\em stellating} means the operation of inserting
a new vertex $v$ inside face $f$ and making it adjacent to all vertices
of $f$. A {\em planar partial 3-tree} is a graph
that is a subgraph of an Apollonian network (see Figure~\ref{graph:10}).
(This definition is different, but equivalent, to the ``standard''
definition of partial 3-trees via treewidth or via chordal supergraphs with
clique-size 4 \cite{BR13}.) We offer the following conjecture:

\begin{conjecture}
    \label{conj:partial3tree}
    Every planar partial 3-tree is 6-coupled-choosable.
\end{conjecture}

\begin{figure}[ht]
    \hspace*{\fill}
    \includegraphics[scale=0.6,page=10]{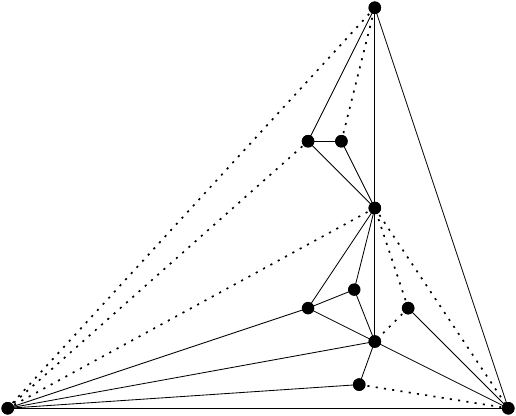}
    \hspace*{\fill}
    \caption{A planar partial 3-tree. Dotted edges show the Apollonian network.}
    \label{graph:10}
\end{figure}

Note that the conjecture holds for Apollonian networks, since these are maximal
planar graphs and these are known to be 6-coupled-choosable 
\cite{coupled-choosability-plane-graphs}.  But this does not
imply 6-coupled-choosability of subgraphs, and so the conjecture remains
open.

Towards the conjecture, we studied several graph classes that are
planar partial 3-trees (and generalize wheels).  
One such class of graphs are the \emph{Halin graphs}, which are defined by starting
with a tree $T$ and adding a cycle between the leaves of $T$. 
See also the solid edges in Figure~\ref{halin}.
Wheel graphs
are the special case of Halin graphs where $T$ is a star graph.
A second class of planar partial 3-trees are the \emph{stellated outer-planar graphs},
obtained by  starting with some outerplanar graph $G$, and stellating
the outer-face.  See also the dashed edges in Figure~\ref{halin}.
Wheel graphs are the special case of stellated outerplanar graphs where
the outerplanar graph is a cycle.
%\begin{lemma}
%    For $n \geq 2$,
%    Halin Graphs are exactly the duals of stellated outerplanar graphs.
%\end{lemma}
%
%We suspect that this result was known before, but have not been able
%to find a reference and therefore provide a proof here.
%
%\begin{proof}
%    Let $G$ be a Halin graph. Every face of $G$ that is not the outer face
%    is adjacent to the outer face. Therefore, the \emph{weak dual}
%    (i.e., the dual graph with the vertex representing the outer face removed)
%    of $G$ is an outerplanar graph. Then adding the outer face and its adjacencies
%    creates a stellated outerplanar graph.
%
%    Let $G$ be a stellated outerplanar graph. The faces of the outerplanar
%    graph form a tree $T$ in the dual. The faces incident to the vertex
%    that stellated the outer face form a cycle $L$, and every such face of $L$
%	shares an edge with some face in $T$.  Hence we can view $L$ as a
%	set of leaves attached to $T$ and then further connected with a cycle,
%	so this is a Halin-graph.
%\end{proof}

One can easily see that Halin graphs
are exactly the duals of stellated outerplanar graphs.
Therefore, any list coupled coloring of
a stellated outerplanar graph corresponds to a list coupled coloring of a Halin graph.
Unfortunately, our upper bound for the coupled choosability of wheel graphs does
not in general extend to Halin graphs.

\begin{theorem}
    \label{counterexample}
    There exists a stellated outerplanar graph
    (equivalently a Halin graph) that is not 5-coupled-colorable
	(in particular therefore it is not 5-coupled-choosable).
\end{theorem}
\begin{proof}
	The Halin-graph $G$ is the triangular prism, see Figure~\ref{halin}
	where we also show the dual graph $G^*$ and the 
	1-planar graph $X(G)$.  The claim holds if we show that there
	is no 5-coloring of the vertices of $X(G)$.
%	We claim that $X(G)$ is not 5-colorable; therefore $G$ is not
%	5-coupled-colorable and in particular not 5-coupled-choosable.

	Assume for contradiction that $X(G)$ had a 5-coloring;  up to symmetry
	we may assume that the triangle formed by the three degree-4-faces
	of $G$ is colored $1,2,3$.  Let $(t,t')$ be the edge that crosses
	the edge colored with $2$ and $3$.  Vertices $t,t'$ are colored 
	with 1, 4 or 5; up to renaming of colors 4 and 5 hence one
	of them is colored 4.
%    Through brute force, it is not overly taxing to confirm that the graph is
%    not 5-coupled-colorable, and hence not 5-coupled-choosable.
%    The coloring of the $K_4$ is unique up to renaming of the colors, this
%    leaves significantly fewer possible colorings to check.

	Starting with this coloring, propagate restrictions on the possible
	colors to other vertices of $X(G)$ along the numerous copies of $K_4$
	(note that all vertices other than $t,t'$ are adjacent to the one
	colored 1).
	This leads to a triangle that has only two possible colors left, a contradiction.
\end{proof}

\begin{figure}[ht]
    \centering
    \includegraphics[scale=0.6,page=11]{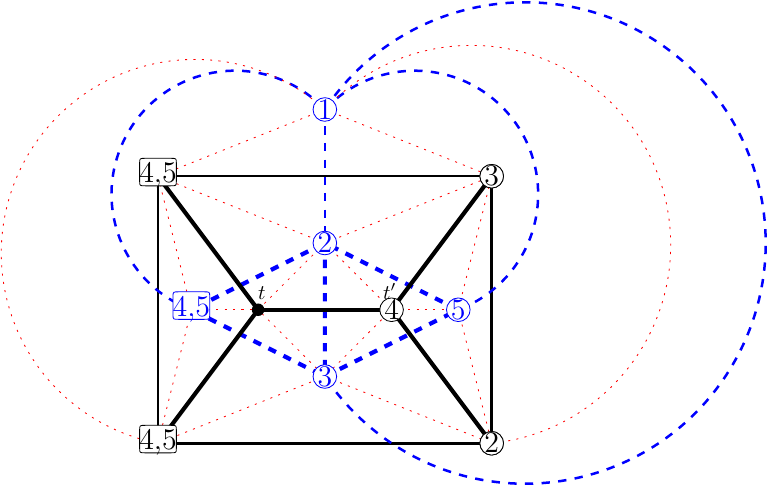}
    \caption{A Halin-graph  $G$ (black solid; the tree is bold),
        and the dual graph $G^*$ (blue dashed) which is
        a stellated outerplanar graph (the outerplanar graph is bold).
	Taking both, and adding the face-vertex incidences
         (red dotted) gives graph $X(G)$.  We also show the only possible 5-coloring
	(up to symmetry) of $X(G)$, which leads
	to a contradiction since a triangle must be colored with
	2 colors.
    }
    \label{halin}
\end{figure}

In particular, this shows that we cannot replace `6' by `5' in
Conjecture~\ref{conj:partial3tree}. We also remark that, in line with
Observation \ref{sub:observation}, a supergraph of the triangular
prism {\em is} 5-coupled-colorable.  Namely, one can insert diagonals
in the degree-4 faces and obtain the octahedron.  An octahedron is
3-colorable because all faces are triangles and the vertex-degrees are
even.  The dual graph (which is the cube) is bipartite and hence 
2-colorable.  Therefore, using disjoint sets of colors for the primal
and the dual graph, we get a 5-coupled-coloring of the octahedron.
%It remains to be shown whether the octahedron
%is 5-coupled-colorable or not.

Returning to wheel graphs, Theorem~\ref{counterexample} shows that
wheels are strictly better (as far as coupled choosability is concerned)
than Halin-graphs.
Now we study a second graph class that lies between the wheels and the
planar partial 3-trees.
These are the \emph{IO-graphs}, which are the planar graphs that can be obtained
by adding an independent set to the interior faces of an outerplanar graph
(see Figure~\ref{iograph}). Certainly any subgraph of a wheel is an IO graph. 

\begin{conjecture}
    \label{conj:IO}
    Every IO-graph is 5-coupled choosable.
\end{conjecture}

We studied subgraphs of wheel graphs because they may be an important 
stepping stone towards Conjecture~\ref{conj:IO}.  In particular, consider
some IO-graph $G$.
obtained from an outerplanar graph $O$ and independent set $I$.
%Each vertex $x \in I$ is placed inside some face of the outerplanar graph $O$.
%We can partition $G$ into subgraphs of wheels by taking for any
%$x \in I$ the vertices on the boundary of the face in which $x$ was placed.
% TB: The above was still not quite correct, some vertices might be taken by many x's.
Let $G^+$ be a maximal IO-graph containing $G$, i.e., add edges to $G$
for as long as the result is simple and an IO-graph.  Then $G^+$ is a
tree of wheels, where each wheel consists of a vertex $x\in I$ with its
neighbours, and  the wheels have been glued together at edges.
Correspondingly $G$ is a tree of subgraphs of wheels.
%In this way, an IO graph can be seen to be a \say{tree} of
%subgraphs of wheels (specifically a 2-clique-sum of subgraphs of wheels).
It may be possible to use Theorem~\ref{main_theorem}
(enhanced with further restrictions on the coloring of some parts)
to prove Conjecture~\ref{conj:IO} by building a coloring
of $G$ incrementally in this tree, 
but this remains future work.
%Finally, our initial goal of determining the list coupled chromatic number
%of planar partial 3-trees remains open.
\begin{figure}[ht]
    \centering
    \includegraphics[scale=0.55,page=6]{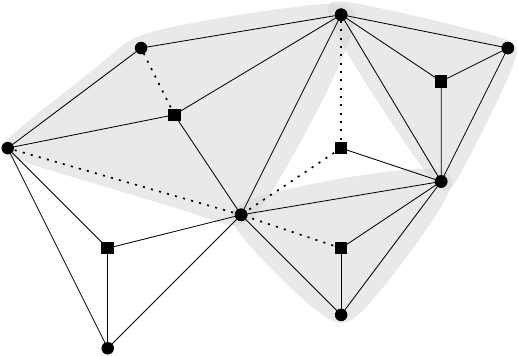}
    \caption{An IO graph $G$ consists of an outerplanar graph (circles)
    and an independent set (squares).  Dotted edges at added to obtain
	$G^+$, and some of the wheels used to build $G^+$ are shaded.}
    \label{iograph}
\end{figure}

We end with some other open questions surrounding list-colorability and
list-coupled-colorability.  Foremost, is every 1-planar graph 7-list-colorable?
Borodin states this to be true \cite{Borodin13}, but quotes the paper by
Wang and Lih \cite{coupled-choosability-plane-graphs} which only deals with
7-coupled-choosability. Hence all optimal 1-planar graphs are 7-list-colorable
but to our knowledge 
the problem remains open for 1-planar graphs that are not subgraphs of
optimal 1-planar graphs (e.g.~any 1-planar graph that contains $K_6$ as
a subgraph).
Second, how easy is it to test whether a planar graph
is $k$-coupled-choosable?  It is known that testing choosability is
$W[1]$-hard with respect to treewidth, but linear-time solvable for constant
treewidth \cite{FFL+11}.  Do these results transfer to coupled-choosability?

%\printbibliography
\bibliographystyle{splncs04}
\bibliography{lifetime}

\end{document}